\documentclass[a4paper,12pt]{article}
\usepackage{amssymb,amsmath,amsthm}
\usepackage{texdraw}
\usepackage{a4wide}
\usepackage{pgf,tikz}
\usetikzlibrary{arrows}

\usepackage{subcaption}

\newtheorem{theorem}{Theorem}[section]

\newtheorem{proposition}[theorem]{Proposition}
\newtheorem{lemma}[theorem]{Lemma}
\newtheorem{corollary}[theorem]{Corollary}

\theoremstyle{definition}

\newtheorem{example}[theorem]{Example}
\newtheorem{remark}[theorem]{Remark}

\begin{document}


\title{Poisson
structures for lifts and periodic reductions of integrable lattice
equations}
\author{Theodoros E. Kouloukas  and Dinh T. Tran
\\
\\Department of Mathematics and Statistics,
\\
La Trobe University, Bundoora VIC 3086, Australia\\
Email: T.Kouloukas@latrobe.edu.au, Dinh.tran@latrobe.edu.au}
\maketitle

\begin{abstract}
We introduce and study suitable Poisson structures for four
dimensional maps derived as lifts and specific periodic reductions
of integrable lattice equations. These maps are Poisson with respect
to these structures and the corresponding integrals are in
involution.

\end{abstract}

\section{Introduction}

Multidimensional consistency (or compatibility) plays an essential
role in the study of partial difference equations on quadrilateral
lattices and can be considered as a criterion related to integrability. In
two dimensions, 3-dimensional (3D) consistency
provides zero curvature representations as well as an effective way
of classifying certain classes of equations \cite{ABS1}. Further
developments in this direction include B{\"a}cklund transformations
of continuous systems, discrete Lagrangian formalism, connection with
Yang-Baxter maps etc.

3D consistent equations give rise to mappings on the lattices by
considering well-posed periodic initial value problems (see e.g.
\cite{Kamp}). In this paper we will focus only on the $(2,2)$
staircase periodic reductions of quad-graph equations \cite{PNC,QCPN} with one field
on each vertex and on the $(1,1)$ periodic reductions of a specific
kind of systems with two fields on the vertices of any elementary
quadrilateral. In this way we always derive four dimensional maps.
From another point of view, 3D consistent equations can be lifted to
four dimensional Yang-Baxter (YB) maps as described in \cite{paptonlift}.
In both cases (periodic reductions and lifts) the corresponding maps
preserve the spectrum of their monodromy matrix. So, from this
spectrum first integrals are obtained. In order to study the
integrability (in the Arnold-Liouville sense) of these maps, one has
to consider a suitable Poisson structure such that the maps are
Poisson and the corresponding integrals in involution. The main
purpose of this paper is to introduce such structures that fulfil
these properties.

We present two different ways of obtaining Poisson structures for
the lifts and the periodic reductions in question. The first one is
from the Sklyanin bracket \cite{skly1}, and it is based on the Lax representation
of the initial equation, while the second one is related to the
existence of a three-leg form \cite{BS}. It seems that we cannot use both ways
in all cases, but in the cases that we can do this the two Poisson
structures coincide.

The special form of the Lax representations of the lifts as YB maps
leads us to an inverse procedure. Instead of starting with a
quad-graph equation and lifting it to a YB map, we can start with a YB
map of a specific form, that admits a Lax pair, and squeeze it down to
an equation (or a system of equations). The latter equation will have
the same Lax pair. We apply this method to some known YB maps
\cite{kp3} and in this way we derive multiparametric versions of two
well-known integrable lattice equations, the cross-ratio and the
lattice nonlinear Schr{\"o}dinger (NLS) system.

The paper is organized as follows. We begin in section 2 by giving
the necessary definitions of 3D consistent equations, YB maps, lifts
and periodic reductions. We also present the multiparametric
cross-ratio equation and its Lax pair. In sections 3 and 4 we study
Poisson structures derived from the Sklyanin bracket and from
three-leg forms respectively. Under some conditions, suitable Poisson
structures for the lifts of quad-graph equations give rise to
suitable Poisson structures for the corresponding $(2,2)$ periodic
reductions and vice versa. In section 5, we consider more general
cases related to a specific kind of systems on quadrilaterals. We
apply our results to the multiparametric NLS system and we derive
the Poisson structure for the lift and $(1,1)$ periodic reduction.
In section 6, we talk about the integrability of the presented maps
and we conclude in section 7 by giving some comments and
perspectives for future work.


\section{Integrable lattice equations and Yang-Baxter maps}

We review some general facts about 3D consistent lattice equations, YB maps and their Lax representations.  3D consistent quadrilateral equations can be lifted to YB maps and YB maps of specific form can be squeezed down to quad-graph equations with the same Lax pair.  

\subsection{3D consistent quad-graph equations and periodic reductions}

We consider equations on quadrilaterals of the type
\begin{equation} \label{eqQ}
Q(w,w_1,w_2,w_{12};\alpha,\beta)=0,
\end{equation}
that can be uniquely solved for any one of their arguments
$w,w_1,w_2,w_{12} \in \mathbb{C}$.

\

\begin{figure}
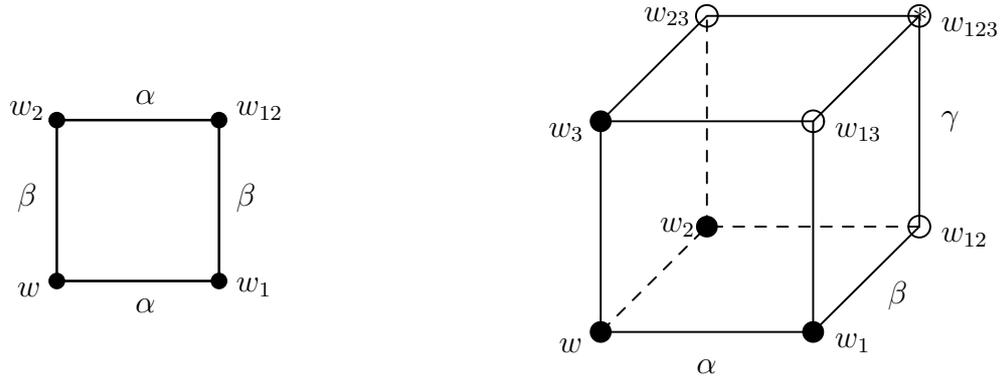

\begin{subfigure}{.5\textwidth}
\centertexdraw{\setunitscale 0.42 \linewd 0.03 \arrowheadtype t:F
\move (-1 -1) \lvec (-1 1) \lvec (1 1) \lvec (1. -1) \lvec(-1 -1)
\move(-1 -1) \fcir f:0.0 r:0.1 \move(-1 1) \fcir f:0.0 r:0.1 \move(1
1) \fcir f:0.0 r:0.1 \move(1 -1) \fcir f:0.0 r:0.1 \htext (-1.5
-1.2) {$w$} \htext (-1.6 1.) {$w_2$} \htext (1.2 1) {$w_{12}$}
\htext (1.2 -1.2) {$w_{1}$} \htext (-0.05 1.2) {$\alpha$} \htext
(-0.05 -1.4) {$\alpha$} \htext (-1.5 -0.15) {$\beta$} \htext (1.2
-0.15) {$\beta$} } \end{subfigure}
\begin{subfigure}{.5\textwidth}
\begin{center}
\centertexdraw{ \setunitscale 1.1 \linewd 0.01

\move (0 0) \linewd 0.01 \lpatt(0.05 0.05) \lvec(0.0 1.0) \lpatt()
\lvec (1.0 1.0) \linewd 0.01 \lvec (1.0 0.0) \lpatt(0.05 0.05)
\lvec(0 0) \lvec(-0.5 -0.5) \lpatt()  \lvec(-0.5 0.5)  \linewd 0.01
\lvec(0 1.0) \move (1. 0)  \linewd 0.01 \lvec(0.5 -0.5) \lvec(-0.5
-0.5) \move(1.0 1.0) \lvec(0.5 0.5) \lvec(-0.5 0.5) \move(0.5 0.5)
\lvec(0.5 -0.5)

\move(0 0) \fcir f:0.0 r:0.05 \move(-0.5 0.5) \fcir f:0.0 r:0.05
\move(0.5 -0.5) \fcir f:0.0 r:0.05 \move(-0.5 -0.5) \fcir f:0.0
r:0.05 \move(1 1) \lcir r:0.05 \move(0.5 0.5) \lcir r:0.05 \move(0
1) \lcir r:0.05 \move(1 0) \lcir r:0.05

\htext (-0.3 0.95) {$w_{23}$} \htext (-0.225 -0.05) {$w_2$} \htext
(1.1 0.9) {$w_{123}$} \htext (-0.75 0.4) {$w_3$}
\htext (0.6 -0.6) {$w_{1}$} \htext (1.1 -0.1) {$w_{12}$} \htext
(-0.7 -0.6) {$w$} \htext (0.6 0.4) {$w_{13}$} \htext (0.85 -0.4)
{$\beta$} \htext (-0.05 -0.7) {$\alpha$} \htext (1.1 0.45)
{$\gamma$} \htext (0.965 0.92){*} }
\end{center}
\end{subfigure}
\caption{Equation (\ref{eqQ}) at the vertices of a quadrilateral and consistency around the cube \label{figcube}}
\end{figure}

By considering the initial values (black points) at the vertices of a cube as in the
Figure \ref{figcube}, 
we can determine the value $w_{123}$ in three different ways. If
all the three values coincide then we call the equation (\ref{eqQ})
{\em{3-dimensional consistent}}.

Adler, Bobenko and Suris (ABS) have classified all the 3D consistent equations on $\mathbb{C}$ with some
extra properties into a list  \cite{ABS1}, commonly referred as the ABS list.

The 3D consistency of a quadrilateral equation gives rise to Lax representations \cite{BS,NW}, i.e. an equation
\begin{equation} \label{laxquad}
L(w_2,w_{12},\alpha)L(w,w_2,\beta)=L(w_1,w_{12},\beta)L(w,w_1,\alpha)
\end{equation}
for some matrix $L$, equivalent to (\ref{eqQ}). The matrix $L$ is called a {\em{Lax matrix}} of equation (\ref{eqQ}).
In a more general setting, quad-graph equations are related to a Lax pair $L$, $M$ that gives rise to a Lax representation 
of the form $ L(w_2,w_{12},\alpha)M(w,w_2,\beta)=M(w_1,w_{12},\beta)L(w,w_1,\alpha)$. In all the cases that we deal with in this paper $L=M$.

Multidimensional maps on quadrilateral lattices are derived from  
3D consistent equations 
by considering well-defined initial value problems. In this paper we will restrict to low
dimensional maps derived by the so called {\em {staircase periodic
initial value problem}} \cite{KQs, PNC, QCPN}. We consider initial values at lattice
points $x_{l,l}$ and $x_{l+1,l}$ which satisfy the periodicity
$x_{l,m}=x_{l+n,m+n}$. By solving the corresponding quad-graph
equation (\ref{eqQ}) at each elementary square of the lattice with
respect to $x_{l,l+1}$, we derive an $n$-dimensional map that maps
the points $x_{l+1,l}$ to the points $x_{l,l+1}$. We will refer to
these maps as $(n,n)$ {\em{staircase periodic reductions}} of the
initial quadrilateral equation.

For each quad-graph equation (\ref{eqQ}), we can define a function  $F:\mathbb{C}^3 \rightarrow \mathbb{C}$, such that (\ref{eqQ})
is equivalent with 
 \begin{equation} \label{ffunction}
w_{2}=F(w_1,w_,w_{12},\alpha,\beta).
\end{equation}
In this way, the map obtained by the $(2,2)$ periodic reduction (Figure \ref{F:two_two_reduction}) can be expressed as 

\begin{equation}
\label{E:reduction}
\mathcal{S}_{\alpha,\beta}:(x_1,x_2,x_3,x_4)\mapsto(x_1^{'}, x_2^{'}, x_3^{'}, x_4^{'}),
\end{equation}
where
\begin{eqnarray}
x_2^{'}&=&F(x_2,x_3,x_1,\alpha,\beta), \ \  x_4^{'}\ = \ F(x_4,x_1,x_3,\alpha,\beta), \label{E:map1} \ \ \ \\
 x_3^{'}&=& F(x_3,x_4^{'},x_2^{'},\alpha,\beta),  \ \
x_1^{'} \ = \ F(x_1,x_2^{'},x_4^{'},\alpha,\beta). \ \ \
\label{E:map2}
\end{eqnarray}

\begin{figure}[h] \label{22Red}
\begin{center}
\begin{tikzpicture}[line cap=round,line join=round,>=triangle 45,x=1.5cm,y=1.5cm][h]
\draw (0.0,1.0)-- (1.0,0.0);
\draw (1.0,0.0)-- (2.0,1.0);
\draw (2.0,1.0)-- (3.0,0.0);
\draw (3.0,0.0)-- (4.0,1.0);
\draw (4.0,1.0)-- (5.0,0.0);
\draw (0.0,1.0)-- (1.0,2.0);
\draw (1.0,2.0)-- (2.0,1.0);
\draw (2.0,1.0)-- (3.0,2.0);
\draw (3.0,2.0)-- (4.0,1.0);
\draw (5.0,0.0)-- (6.0,1.0);
\draw (4.0,1.0)-- (5.0,2.0);
\draw (5.0,2.0)-- (6.0,1.0);
\draw (1.0,2.0)-- (2.0,3.0);
\draw (2.0,3.0)-- (3.0,2.0);
\draw (3.0,2.0)-- (4.0,3.0);
\draw (4.0,3.0)-- (5.0,2.0);
\draw [->] (1.0,0.0) -- (1.0,2.0);
\draw [->] (2.0,1.0) -- (2.0,3.0);
\draw [->] (3.0,0.0) -- (3.0,2.0);
\draw [->] (4.0,1.0) -- (4.0,3.0);
\draw [fill=black] (0.0,1.0) circle (2pt);
\draw(-0.1,1.2085663160937274) node {$x_1$};
\draw [fill=black] (1.0,0.0) circle (2pt);
\draw(1.1115929156639173,-0.21046306447702798) node {$x_4$};
\draw (0.39235541979040217,0.5043129347175777) node {$\alpha$};
\draw [fill=black] (2.0,1.0) circle (2pt);
\draw (2.0005444724900006,1.2085663160937274) node {$x_3$};
\draw (1.7109574955585132,0.4593605912254829) node {$\beta$};
\draw [fill=black] (3.0,0.0) circle (2pt);
\draw (3.1044801438134484,-0.21046306447702798) node {$x_2$};
\draw(2.3702585334425685,0.4593605912254829) node {$\alpha$};
\draw [fill=black] (4.0,1.0) circle (2pt);
\draw(4.008415815136897,1.2085663160937274) node {$x_1$};
\draw (3.7038447237080443,0.4593605912254829) node {$\beta$};
\draw [fill=black] (5.0,0.0) circle (2pt);
\draw (5.112351486460345,-0.2046306447702798) node {$x_4$};
\draw (4.3631457615921,0.4593605912254829) node {$\alpha$};
\draw [fill=black] (1.0,2.0) circle (2pt);
\draw (0.91115929156639173,2.2725019874171747) node {$x_4^{'}$};
\draw (0.707021824235065,1.4632962625489305) node {$\beta$};
\draw (1.3663228621191206,1.4632962625489305) node {$\alpha$};
\draw [fill=black] (3.0,2.0) circle (2pt);
\draw (2.981044801438134484,2.3125019874171747) node {$x_2^{'}$};
\draw (2.6999090523845966,1.4632962625489305) node {$\beta$};
\draw(3.3741942047660167,1.4632962625489305) node {$\alpha$};
\draw [fill=black] (6.0,1.0) circle (2pt);
\draw(6.101303043286428,1.2085663160937274) node {$x_3$};
\draw (5.711716066354941,0.4593605912254829) node {$\beta$};
\draw [fill=black] (5.0,2.0) circle (2pt);
\draw (5.112351486460345,2.2125019874171747) node {$x_4^{'}$};
\draw (4.7077803950314925,1.4632962625489305) node {$\beta$};
\draw (5.367081432915548,1.4632962625489305) node {$\alpha$};
\draw [fill=black] (2.0,3.0) circle (2pt);
\draw (2.1005444724900006,3.216437658740622) node {$x_3^{'}$};
\draw(1.7109574955585132,2.452247819375013) node {$\beta$};
\draw(2.3702585334425685,2.452247819375013) node {$\alpha$};
\draw [fill=black] (4.0,3.0) circle (2pt);
\draw (4.108415815136897,3.216437658740622) node {$x_1^{'}$};
\draw(3.7038447237080443,2.452247819375013) node {$\beta$};
\draw(4.3631457615921,2.452247819375013) node {$\alpha$};
\end{tikzpicture}
\caption{The $(2,2)$ staircase periodic reduction \label{F:two_two_reduction}}
\end{center}
\end{figure}
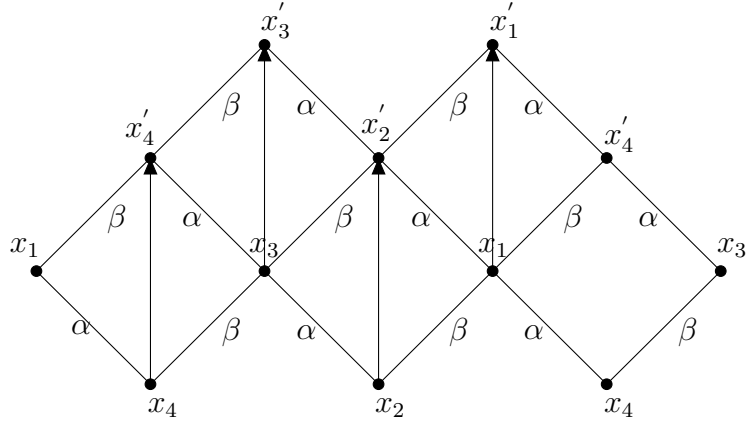

\subsection{Yang-Baxter maps}

A {\it{Yang-Baxter map}} is a map $R: \mathcal{X} \times
\mathcal{X} \rightarrow \mathcal{X} \times \mathcal{X}$,
$R:(x,y)\mapsto (u(x,y),v(x,y))$, that satisfies the {\em
Yang-Baxter equation} \cite{Baxter,Drin,ves2,Yang}
\begin{equation*}\label{YBprop}
 R_{23}\circ R_{13}\circ
R_{12}=R_{12}\circ R_{13}\circ R_{23}.
\end{equation*} Here by $R_{ij}$
for $i,j=1,2,3$, we denote the action of the map $R$ on the $i$ and
$j$ factor of $\mathcal{X} \times \mathcal{X} \times \mathcal{X}$,
i.e. $R_{12}(x,y,z)=(u(x,y),v(x,y),z)$,
$R_{13}(x,y,z)=(u(x,z),y,v(x,z))$ and
$R_{23}(x,y,z)=(x,u(y,z),v(y,z))$. A YB map $R:(\mathcal{X} \times
\mathcal{I}) \times (\mathcal{X} \times \mathcal{I}) \mapsto
(\mathcal{X} \times \mathcal{I}) \times (\mathcal{X} \times
\mathcal{I})$, with
\begin{equation} \label{pYB}
R:((x,\alpha),(y,\beta))\mapsto((u,\alpha),(v,\beta))= ((u(x,\alpha,
y,\beta),\alpha),(v(x,\alpha, y,\beta),\beta)),
\end{equation}
is called a \emph{parametric YB map} (\cite{ves2,ves3}). We usually
keep the parameters separately and denote (\ref{pYB}) as
$R_{\alpha,\beta}(x,y):\mathcal{X}\times \mathcal{X} \rightarrow
\mathcal{X} \times \mathcal{X}$. Generally, $\mathcal{X}$ can be
any set. From our point of view, the sets $\mathcal{X}$ and
$\mathcal{I}$ have the structure of an algebraic variety and the
maps that we consider are birational.

According to \cite{ves4}, a {\em Lax matrix} of the parametric YB
map (\ref{pYB}) is a matrix $L$ that depends on a point $x \in \mathcal{X}$,  
a parameter $\alpha \in \mathcal{I}$ and a spectral parameter $\zeta\in \mathbb{C}$, such
that
\begin{equation} \label{laxmat}
L(u,\alpha,\zeta)L(v,\beta,\zeta)=L(y,\beta,\zeta)L(x,\alpha,\zeta).
\end{equation}
Furthermore, if equation (\ref{laxmat}) is equivalent to 
$(u,v)=R_{\alpha,\beta}(x,y)$ then $L$ is called {\em
strong Lax matrix}. We often omit the spectral parameter $\zeta$ and denote the Lax matrix $L(x,\alpha,\zeta)$ just by $L(x,\alpha)$.
The next proposition, presented in \cite{kp1}
and essentially also in \cite{ves2}, provides a sufficient condition for
solutions of equation (\ref{laxmat}) to satisfy the YB equation.

\begin{proposition} \label{3fact}
If $u=u_{\alpha,\beta}(x,y), v=v_{\alpha,\beta}(x,y) $ satisfy
(\ref{laxmat}), for a matrix $L$  and the equation $$L(
\hat{x},\alpha )L( \hat{y} ,\beta )L(\hat{z}, \gamma )= L(x
,\alpha)L(y, \beta)L(z, \gamma)$$ implies that $\hat{x}=x, \
\hat{y}=y$ and $\hat{z}=z$, for every $x,y,z \in \mathcal{X}$, then
$R_{\alpha,\beta}(x,y)\mapsto(u,v)$ is a parametric YB map with Lax
matrix $L$.
\end{proposition}

The dynamical aspects of YB maps have been studied by Veselov
\cite{ves2,ves3}. For any YB map there is a hierarchy of commuting
transfer maps which preserve the spectrum of the corresponding
monodromy matrix.
Furthermore, YB maps and their Lax matrices are related to B{\"a}cklund transformations of 
continuous and discrete integrable systems (see e.g. \cite{sokor,skly3}).

\subsection{Lift of 3D consistent quad-graph equations to Yang-Baxter 
Maps } \label{secLift}

Three dimensional consistent equations on quad-graphs can be lifted
to YB maps. This lifting procedure has been described in
\cite{paptonlift}.

Let $Q(w,w_1,w_2,w_{12},\alpha,\beta)=0$ be a quad-graph 
equation affine linear in each argument
$w,w_1,w_2,w_{12}$. As before, we define the function $F$ by solving
this equation with respect to $w_{2}$, by
$w_{2}=F(w_1,w_,w_{12},\alpha,\beta).$
The lift of the equation $Q$ is defined as the map
$R_{\alpha,\beta}(x_1,x_2,y_1,y_2) \mapsto (u_1,u_2,v_1,v_2)$, with
\begin{equation} \label{liftYB}
u_1=F(y_1,x_1,y_2,\alpha,\beta), \ u_2=y_2, \ v_1=x_1, \
v_2=F(x_2,x_1,y_2,\alpha,\beta).
\end{equation}

As was shown in \cite{paptonlift}, under some additional conditions, $R_{\alpha,\beta}$ is a YB map. 
All the lifts of the equations of the ABS list are YB maps.
The proof of the YB property of these maps follows from the 3D
consistency of the initial equations.

The name lift of the quad-graph equation is
justified from the following observation. If we set $x_2=y_1$, then
$u_1=v_2$, and by labeling the variables as $x_2=y_1=w_1$,
$v_1=x_1=w$, $u_2=y_2=w_{12}$ and $u_1=v_2=w_{2}$, both first and
last equations of (\ref{liftYB}) reduce to the initial quad-graph
equation $Q(w,w_1,w_2,w_{12},\alpha,\beta)=0$. Having in mind this observation, we can easily prove the next proposition.

\begin{proposition}\label{propinverse}
If $L$ is a Lax matrix of the lift of a quad-graph equation then $L$
is also a Lax matrix of the quad-graph equation.
\end{proposition}

\begin{proof}
Let $L$ be a Lax matrix of the lift of a quad-graph equation $w_2=F(w_1,w,w_{12},\alpha,\beta)$.
Then
\begin{equation} \label{spr}
L(u_1,u_2,\alpha)L(v_1,v_2,\beta)=L(y_1,y_2,\beta)L(x_1,x_2,\alpha),
\end{equation}
for $u_1$, $u_2$, $v_1$ and $v_2$ given by (\ref{liftYB}). By setting
$x_2=y_1= w_1$,  $v_1=x_1=w$, $u_2=y_2=w_{12}$, we have $u_1=v_2=w_{2}$. If we substitute 
these to (\ref{spr}), we derive the Lax representation (\ref{laxquad}). \end{proof}

In many cases, including the examples that will be presented here, the converse of this proposition also holds. That means that
the Lax matrix of the initial equation is a Lax matrix of its lift (not always strong). In this cases, the YB property  of the
map (\ref{liftYB}) can be proved from  Proposition \ref{3fact}.

Using Proposition \ref{propinverse} we can reverse the procedure and
derive 3D consistent systems from YB maps of the form (\ref{liftYB})
as the next example shows.

\subsubsection*{Multiparametric cross-ratio equation}
The map
$R_{\bar{\alpha},\bar{\beta}}(x_1,x_2,y_1,y_2)=(F(y_1,x_1,y_2,\bar{\alpha},\bar{\beta}),y_2,x_1,F(x_2,x_1,y_2,\bar{\alpha},\bar{\beta})),$
with
$$F(x,x_1,x_2,\bar{\alpha},\bar{\beta})=
\frac{\beta \alpha_2 x_2(\alpha_1 x_1-\alpha_2 x)x_2 + \alpha
\beta_1 x_1(\beta_1x-\beta_2 x_2)} { \beta \alpha_1(\alpha_1
x_1-\alpha_2 x)+ \alpha \beta_2 (\beta_1x-\beta_2 x_2)}$$ and vector
parameters $\bar{\alpha}=(\alpha,\alpha_1,\alpha_2)$ and
$\bar{\beta}=(\beta,\beta_1,\beta_2)$ is a YB map\footnote{Case I YB map of the classification in \cite{kp3}, after a change of variables.} with Lax matrix
\begin{equation} \label{laxGcr}
L(x_1,x_2,\bar{\alpha})=\left(
\begin{array}{cc}
\alpha_1 \zeta +\frac{\alpha x_2}{\alpha_1 x_1-\alpha_2 x_2}&-\frac{\alpha x_1 x_2}{\alpha_1 x_1-\alpha_2 x_2}\\
\frac{\alpha}{\alpha_1 x_1-\alpha_2 x_2} & \alpha_2
\zeta-\frac{\alpha x_1}{\alpha_1 x_1-\alpha_2 x_2} \end{array}
\right).
\end{equation}
According to Proposition \ref{propinverse}, the quad-graph equation
$w_{2}=F(w_1,w,w_{12},\bar{\alpha},\bar{\beta})$ satisfies the Lax
equation (\ref{laxquad}). This equation can be seen as a 
multiparametric version of the cross-ratio equation (special case of equation (33) in \cite{HiVi}) and is given as follows
\begin{equation} \label{genQ1}
\alpha(\beta_1 w-\beta_2 w_2)(\beta_1 w_1-\beta_2 w_{12})-
\beta(\alpha_1 w-\alpha_2 w_1)(\alpha_1 w_2-\alpha_2 w_{12})=0.
\end{equation}

This equation satisfies the 3D consistency and the
tetrahedron properties \cite{ABS1}. Therefore, one can derive its
Lax pair which contains $3$ spectral parameters. For
$\alpha_1=\alpha_2=\beta_1=\beta_2=1$, equation (\ref{genQ1}) is
reduced to $Q_1$ with $\delta=0$ of ABS classification
(cross-ratio equation).
However, equation (\ref{genQ1}) does not possess the symmetries of
the square ($D_4$-symmetry) in the normal way, but it satisfies the
following symmetry property
$$Q(w,w_1,w_2,w_{12},\alpha,\alpha_1,\alpha_2,\beta,\beta_1,\beta_2)=-Q(w_{12},w_1,w_2,w,\beta,\beta_2,\beta_1,\alpha,\alpha_2,\alpha_1).$$
Also, by using the non-autonomous transformation $w_{l,m}\mapsto
A^{l-l_0}B^{m-m_0}w_{l,m}$, where $A=\alpha_1/\alpha_2$ and
$B=\beta_1/\beta_2$, equation \eqref{genQ1} is brought to the cross-ratio equation.
The map $R_{\bar{\alpha},\bar{\beta}}$ is a Poisson map with respect
to the Sklyanin bracket. We will give the corresponding Poisson
structure in the next section (example \ref{exGcr}).

\section{Poisson structure derived from the Sklyanin bracket}
In this section, we derive Poisson structures for the lifts and $(2,2)$ periodic reductions 
from the Sklyanin bracket \cite{skly1} that is related to their Lax representations. 

\subsection{The Sklyanin bracket on the lifts of quad-graph equations}

In some cases, Lax pairs of
YB maps are derived by reduction of polynomial matrices to
the symplectic leaves of the Sklyanin bracket and the
YB maps turn out to be symplectic with respect
to this bracket \cite{kp1,kp3}. Furthermore, the Sklyanin bracket
ensures that the integrals that we derive from the trace of the
corresponding monodromy matrix of periodic initial value problems
will be in involution \cite{bab,ves2,ves3}.

 We denote by $\mathbb{L}_{m}^n$ the set
of $m \times m$, $n-$degree polynomial matrices of the form
$$L(\mathbf{x},\zeta)=X_0+\zeta X_1+...+\zeta ^{n} X_n,$$
here $\mathbf{x}=(X_0,...,X_n)$, $ X_i \in Mat_{m \times
m}$ and $\zeta \in \mathbb{C}$. The functions that
depend on the coefficients $X_i$ form a Poisson algebra with respect
to the \textit{r-matrix quadratic Poisson bracket} or
\textit{Sklyanin bracket}, which in tensor notation is given by the
formula
\begin{equation}
\{L(\mathbf{x},\zeta ) \ \overset{\otimes }{,} \
L(\mathbf{x},h)\}=[\mathbf{r}(\zeta,h),L(\mathbf{x},\zeta )\otimes
L(\mathbf{x},h)]. \label{sklyanin}
\end{equation}
Generally, $\mathbf{r}(\zeta,h)$ is a classical $r$-matrix, i.e. a
solution of the classical YB equation \cite{bab,bel,rey,skly1, skly2}. We will consider the simple case
$\mathbf{r}(\zeta,h)=\frac{r}{\zeta-h}$, where $r$ denotes the
permutation matrix: $r(x\otimes y) = y\otimes x$.

The $m^2$ elements of the highest degree term $X_n$ and the $mn$ coefficients of the determinant of $L(\mathbf{x},\zeta )$
are Casimir functions. By restricting to the common level set of the Casimir functions, we derive $m^2(n+1)$ dimensional
matrices that satisfy (\ref{sklyanin}).

The Sklyanin bracket can be extended to the Cartesian product $\mathbb{L}_{m}^n \times \mathbb{L}_{m}^n$ in the natural
way by setting
\begin{eqnarray}
\{L(\mathbf{x},\zeta ) \ \overset{\otimes }{,} \
L(\mathbf{x},h)\}&=&[\mathbf{r}(\zeta,h),L(\mathbf{x},\zeta )\otimes
L(\mathbf{x},h)], \nonumber  \\
 \{L(\mathbf{y},\zeta ) \ \overset{\otimes }{,}
\ L(\mathbf{y},h)\}&=&[\mathbf{r}(\zeta,h),L(\mathbf{y},\zeta
)\otimes L(\mathbf{y},h)], \nonumber \\
 \{L(\mathbf{x},\zeta ) \ \overset{\otimes }{,} \ L(\mathbf{y},h)\} &=& 0,  \label{extskl}
\end{eqnarray}
for $\ (L(\mathbf{x},\zeta ),L(\mathbf{y},\zeta)) \in \mathbb{L}_{m}^n
\times \mathbb{L}_{m}^n$.

\begin{proposition} \label{inverse}
Let $R_{\alpha,\beta}:(\mathbf{x},\mathbf{y}) \mapsto (\mathbf{u},\mathbf{v})$ be a YB map with Lax matrix $L$ that satisfies (\ref{extskl}). Then
\begin{equation}
\{ L(\mathbf{u},\alpha,\zeta ) L(\mathbf{v},\beta,\zeta )  \overset{\otimes }{,}
L(\mathbf{u},\alpha,h) L(\mathbf{v},\beta,h ) \}=[\mathbf{r}(\zeta,h), L(\mathbf{u},\alpha,\zeta ) L(\mathbf{v},\beta,\zeta )\otimes
L(\mathbf{u},\alpha,h) L(\mathbf{v},\beta,h )]. \label{pom}
\end{equation}
\end{proposition}

\begin{proof}
From equations (\ref{extskl}) we derive
\begin{equation*}
\{ L(\mathbf{y},\beta,\zeta ) L(\mathbf{x},\alpha,\zeta )  \overset{\otimes }{,}
L(\mathbf{y},\beta,h) L(\mathbf{x},\alpha,h ) \}=[\mathbf{r}(\zeta,h), L(\mathbf{y},\beta,\zeta ) L(\mathbf{x},\alpha,\zeta )\otimes
L(\mathbf{y},\beta,h) L(\mathbf{x},\alpha,h )]. \footnote{This is a well known  property of the Sklyanin bracket, often referred to as the \it{comultiplication} property.}
\end{equation*}
Using (\ref{laxmat}), we have
\begin{eqnarray*}
\{ L(\mathbf{u},\alpha,\zeta ) L(\mathbf{v},\beta,\zeta )  \overset{\otimes }{,}
L(\mathbf{u},\alpha,h) L(\mathbf{v},\beta,h ) \}= \{ L(\mathbf{y},\beta,\zeta ) L(\mathbf{x},\alpha,\zeta )  \overset{\otimes }{,}
L(\mathbf{y},\beta,h) L(\mathbf{x},\alpha,h ) \} \\
= [\mathbf{r}(\zeta,h), L(\mathbf{y},\beta,\zeta ) L(\mathbf{x},\alpha,\zeta )\otimes
L(\mathbf{y},\beta,h) L(\mathbf{x},\alpha,h )]  \ \\
= [\mathbf{r}(\zeta,h), L(\mathbf{u},\alpha,\zeta ) L(\mathbf{v},\beta,\zeta )\otimes
L(\mathbf{u},\alpha,h) L(\mathbf{v},\beta,h )].
\end{eqnarray*}
\end{proof}

 Equation (\ref{pom}) is a necessary condition for a YB map to be Poisson
 with respect to the bracket (\ref{extskl}). In many cases, the Poisson property of these maps follows from the uniqueness of the refactorization of the Lax matrices (\cite{kp1,kp3}).

\begin{example} \label{liftKdVex}

The map
\begin{equation} \label{liftKdV}
R_{\alpha,\beta}(x_1,x_2,y_1,y_2)=(y_1+ \frac{\alpha -
\beta}{x_1-y_2},y_2,x_1,x_2+ \frac{\alpha - \beta}{x_1-y_2})
\end{equation}
 is a parametric YB map with (not strong) Lax matrix
 $$L(x_1,x_2,\alpha)=\left(
\begin{array}{cc}
 x_1 & \alpha+x_1 x_2-\zeta \\
 -1 & -x_2
\end{array}
\right), $$
This map was derived
in \cite{kp1} and can be considered as a lift of KdV quad-graph
equation ($H_1$ equation of the ABS classification list \cite{ABS1})
$$(w_{12}-w)(w_1-w_2)=\alpha-\beta.$$
We can verify that equations (\ref{extskl}) are equivalent to 
$$ \ \{x_1,x_2 \}=1, \ \{y_1,y_2\}=1, \ \{x_i,y_j\}=0, \  \text{for} \ i,j=1,2,$$
and that the YB map (\ref{liftKdV}) is Poisson with respect to this bracket.
\end{example}

\begin{example} \label{exGcr}
The Lax matrix of the multiparametric cross-ratio equation
(\ref{laxGcr}), presented in the previous section, satisfies the
Sklyanin bracket. In this case the extended Sklyanin bracket
(\ref{extskl}) is equivalent to 
\begin{equation} \label{Gcrbr}
\ \{x_1,x_2 \}=-\frac{(\alpha_1 x_1- \alpha_2 x_2)^2}{\alpha},
\{y_1,y_2\}=-\frac{(\beta_1 y_1- \beta_2 y_2)^2}{\beta},
\{x_i,y_j\}=0.
\end{equation}
The corresponding lift of this equation $R_{\bar{\alpha},\bar{\beta}}$ is Poisson
with respect to this bracket.
\end{example}

\subsection{The Sklyanin bracket on $(2,2)$ staircase periodic
reductions}

Next, we establish a connection between the lifts and the $(2,2)$ periodic reductions of quad-graph equations. 
In this way, and under some additional conditions, the Poisson structure of the lift of a quad-graph equation gives rise to a suitable Poisson structure for the periodic reduction. We begin with the following lemma.

We consider a function $F:\mathbb{I} \subset \mathbb{C}^5 \rightarrow \mathbb{C}$, such that
 \begin{equation} \label{D4sy}
 F(x,x_1,x_2,\alpha,\beta)=F(x,x_2,x_1,\beta,\alpha),
\end{equation}
as well as the three parametric maps
\begin{eqnarray}
R_{\alpha,\beta}(x_1,x_2,y_1,y_2) = (F(y_1,x_1,y_2,\alpha,\beta),y_2,x_1,F(x_2,x_1,y_2,\alpha,\beta)) :=(u_1,u_2,v_1,v_2) \label{R}\\
T^1_{\alpha,\beta}(x_1,x_2,y_1,y_2) = (F(y_1,x_1,y_2,\alpha,\beta),y_2,x_1,F(x_2,y_2,x_1,\alpha,\beta)) :=(u_1',u_2',v_1',v_2') \label{S}\\
T^2_{\alpha,\beta}(x_1,x_2,y_1,y_2) =
(F(y_1,y_2,x_1,\alpha,\beta),y_2,x_1,F(x_2,x_1,y_2,\alpha,\beta))
:=(\tilde{u}_1,\tilde{u}_2,\tilde{v}_1,\tilde{v}_2) \label{T}
\end{eqnarray}

\begin{lemma} \label{lemma1}
 Let $R_{\alpha,\beta}$ be a Poisson map with respect to a Poisson structure of the form
\begin{equation} \label{poisson1}
\pi_1=J_{\alpha,\beta}(x_1,x_2)\frac{\partial}{\partial x_1} \wedge
\frac{\partial}{\partial
x_2}+J_{\beta,\alpha}(y_1,y_2)\frac{\partial}{\partial y_1} \wedge
\frac{\partial}{\partial y_2},
\end{equation}
 then the map
$\mathbf{S}_{\alpha,\beta}=T^2_{\alpha,\beta} \circ T^1_{\alpha,\beta}$
is Poisson with respect to
\begin{equation} \label{poisson2}
\pi_2=J_{\beta,\alpha}(x_1,x_2)\frac{\partial}{\partial x_1} \wedge
\frac{\partial}{\partial
x_2}-J_{\beta,\alpha}(y_1,y_2)\frac{\partial}{\partial y_1} \wedge
\frac{\partial}{\partial y_2},
\end{equation}
 if and only if
\begin{equation} \label{condition}
J_{\beta,\alpha}(x_1,x_2)\frac{\partial u_1'}{\partial x_1}  \frac{\partial v_2'}{\partial x_2}=J_{\beta,\alpha}(y_1,y_2)\frac{\partial u_1'}{\partial y_1}  \frac{\partial v_2'}{\partial y_2}.
\end{equation}
\end{lemma}

In the case that the function $F$ is defined by a quad-graph 
equation $Q(w,w_1,w_2,w_{12})$  from (\ref{ffunction}), then the
map $R_{\alpha,\beta}$ is the lift of this equation and
the map
\begin{equation} \label{perrduction}
\mathcal{S}_{\alpha,\beta}(x_1,x_2,x_3,x_4)=\mathbf{S}_{\alpha,\beta}(x_1,x_2,x_4,x_3)
\end{equation}
is its $(2,2)$ staircase periodic reduction (\ref{E:reduction}).
In the coordinates $(x_1,x_2,x_3,x_4)$, the corresponding Poisson structure $\pi_2$ of the last lemma becomes
\begin{equation} \label{poisson3}
\pi_2=J_{\beta,\alpha}(x_1,x_2)\frac{\partial}{\partial x_1} \wedge
\frac{\partial}{\partial
x_2}+J_{\beta,\alpha}(x_3,x_4)\frac{\partial}{\partial x_3} \wedge
\frac{\partial}{\partial x_4}
\end{equation}
and condition (\ref{condition}) can be written as
$\pi_2(dx_2^{'},dx_4^{'})=0$, where $x_2^{'},x_4^{'}$ are given in
(\ref{E:map1}).

As we saw in the previous section, in many cases the lift of a
quad-graph equation is Poisson with respect to the extended Sklyanin
bracket. Obviously the bracket (\ref{extskl}) is of the form of
(\ref{poisson1}) for some function $J_{\alpha,\beta}$ (that depends
only in one parameter). In these cases, if the conditions of lemma
\ref{lemma1} are fulfilled, a modification of the Sklyanin bracket
from $\pi_1$ to $\pi_2$ gives rise to a suitable Poisson structure
for the $(2,2)$ periodic reduction of the initial quad-graph
equation.

Finally, we have to remark that the property (\ref{D4sy}), follows
from the following symmetry of the square
\begin{equation} \label{sqsym}
Q(w,w_1,w_2,w_{12},\alpha,\beta)=\pm
Q(w_{12},w_1,w_2,w,\beta,\alpha).
\end{equation}

We can summarise all these facts in the following proposition.

\begin{proposition} \label{proplem}
We consider a quad-graph equation
$Q(w,w_1,w_2,w_{12},\alpha,\beta)$, that satisfies the symmetry
condition (\ref{sqsym}), with corresponding lift $R_{\alpha,\beta}$
a Poisson YB map with respect to the Sklyanin bracket (\ref{extskl})
and Lax matrix $L$. The $(2,2)$ periodic reduction
$\mathcal{S}_{\alpha,\beta}:(x_1,x_2,x_3,x_4)\mapsto(x_1^{'},
x_2^{'}, x_3^{'}, x_4^{'})$ is Poisson with respect to the bracket
defined by
\begin{eqnarray*}
\{L(x_1,x_2,\beta,\zeta ) \ \overset{\otimes }{,} \ L(x_1,x_2,\beta
,h)\}_2&=&[\mathbf{r}(\zeta,h),L(x_1,x_2,\beta,\zeta ) \otimes
L(x_1,x_2,\beta,h)], \nonumber  \\
\{L(x_3,x_4,\beta,\zeta ) \ \overset{\otimes }{,} \ L(x_3,x_4,\beta
,h)\}_2&=&[\mathbf{r}(\zeta,h),L(x_3,x_4,\beta,\zeta ) \otimes
L(x_3,x_4,\beta,h)]
\end{eqnarray*}
and $\{ x_1,x_3 \}_2=\{ x_1,x_4 \}_2=\{ x_2,x_3 \}_2=\{ x_2,x_4
\}_2=0$, if and only if $\{x_2^{'},x_4^{'} \}_2=0$.
\end{proposition}

\begin{example}
We consider the lift of KdV quad-graph equation of  example  \ref{liftKdVex}
$$R_{\alpha,\beta}(x_1,x_2,y_1,y_2) = (F(y_1,x_1,y_2,\alpha,\beta),y_2,x_1,F(x_2,x_1,y_2,\alpha,\beta)) ,$$
where $F(x,x_1,x_2,\alpha,\beta)=x+ \frac{\alpha -\beta}{x_1-x_2}$.
As we mentioned before, this is a Poisson map with respect to the
Poisson structure (\ref{poisson1}), for
$J_{\alpha,\beta}(x_1,x_2)=1$ and satisfies the condition
(\ref{condition}). So, the corresponding $(2,2)$ staircase periodic
reduction $\mathcal{S}_{\alpha,\beta}$ (\ref{perrduction}) is
Poisson with respect to the bracket $\pi_2=\frac{\partial}{\partial
x_1} \wedge \frac{\partial}{\partial x_2}+\frac{\partial}{\partial
x_3} \wedge \frac{\partial}{\partial x_4}.$
\end{example}

\begin{example}
We consider the 3D consistent quad-graph equation
\begin{equation} \label{Q1}
\alpha (w-w_2)(w_1-w_{12})-\beta(w-w_1)(w_2-w_{12})=0.
\end{equation}
This is the $Q_1$ quad-graph equation of the ABS list for $\delta=0$.
The lift of the equation gives rise to the YB map
$$R_{\alpha,\beta}(x_1,x_2,y_1,y_2) = (F(y_1,x_1,y_2,\alpha,\beta),y_2,x_1,F(x_2,x_1,y_2,\alpha,\beta)) ,$$
with $$F(x,x_1,x_2,\alpha,\beta)=\frac{\alpha x_1(x-x_2)+\beta x_2(x_1-x)}{\alpha(x-x_2)+\beta(x_1-x)}$$
and (strong) Lax matrix
 $$L(x_1,x_2,\alpha)=\left(
\begin{array}{cc}
 \zeta +\frac{\alpha x_2}{x_1-x_2} &-\frac{\alpha  x_1 x_2}{x_1-x_2} \\
\frac{\alpha }{x_1-x_2}  &\zeta- \frac{\alpha x_1}{x_1-x_2}
\end{array}
\right). $$ In this case, the corresponding Sklyanin bracket
(\ref{extskl}) is equivalent to 
$$ \ \{x_1,x_2 \}=-\frac{(x_1-x_2)^2}{\alpha}, \ \{y_1,y_2\}=-\frac{(y_1-y_2)^2}{\beta}, \ \{x_i,y_j\}=0, \  \text{for} \ i,j=1,2,$$
i.e. the Poisson structure (\ref{poisson1}) for
$J_{\alpha,\beta}(x_1,x_2)=-\frac{(x_1-x_2)^2}{\alpha}$. The YB map
$R_{\alpha,\beta}$ is Poisson with respect to this structure  and
the condition (\ref{condition}) is satisfied. So, the corresponding
periodic reduction (\ref{perrduction})
is Poisson with respect to
$$\pi_2=-\frac{(x_1-x_2)^2}{\beta}\frac{\partial}{\partial x_1} \wedge
\frac{\partial}{\partial
x_2}-\frac{(x_3-x_4)^2}{\beta}\frac{\partial}{\partial x_3} \wedge
\frac{\partial}{\partial x_4}.$$

\end{example}


We have to remark that there are cases where the lift of a
quad-graph equation is Poisson with respect to the Sklyanin bracket 
but the $(2,2)$ periodic reduction is not Poisson with respect to
the corresponding bracket of Proposition \ref{proplem}, because some
of the conditions are not satisfied. For example, the lift of the
multiparametric cross-ratio equation (\ref{genQ1}) is a Poisson map
with respect to the Sklyanin bracket (\ref{Gcrbr}), but this
equation does not satisfy the symmetry condition (\ref{sqsym}) and
the corresponding $(2,2)$ periodic reduction is not Poisson with
respect to the bracket of Proposition \ref{proplem}.

Furthermore, we cannot define a Poisson bracket from 
(\ref{extskl}) for any Lax representation of a lift of a quad-graph
equation. In the cases that the quad-graph equations have an
equivalent three-leg form (e.g. the equations of the ABS list), we can
derive directly a Poisson structure for the $(2,2)$ periodic
reductions and subsequently (except for the $Q_4$ case) a Poisson
structure for the corresponding lifts. We study these cases in the
next section.

\section{Poisson structure derived from three-leg forms}
In this section, we derive Poisson brackets for the $(2,2)$ periodic reduction (Figure \ref{F:two_two_reduction}) of the ABS equations from the so-called three-leg forms.

Recall that a three-leg form of equation \eqref{eqQ} centred at $f$ is defined as follows
\begin{equation}
\label{E:3leg_form}
\psi(w,w_1,\alpha)-\psi(w,w_2,\beta)=\phi(w,w_{12},\alpha,\beta),
\end{equation}
for some functions $\psi$ and $\phi$ such that equation
\eqref{E:3leg_form} is equivalent to equation~\eqref{eqQ}. In this
equation, $\psi$ is called a short leg and $\phi$  a long leg. It
is noted that all equations in the ABS list \cite{ABS1} (after some
transformations) have  a  three leg-form. Due to the $D_4$ symmetry
properties of these equations,  their three-leg forms can be
centred at any vertex of the quadrilateral.
\begin{proposition}
\label{P:symplectics}
The $(2,2)$ periodic reduction map $\mathcal{S}_{\alpha,\beta}:(x_1,x_2,x_3,x_4)\mapsto(x_1^{'}, x_2^{'}, x_3^{'}, x_4^{'})$ of the ABS equations, 
described by  \eqref{E:map1} and \eqref{E:map2}, satisfies 
\begin{equation}
\label{E:sympletic}
s(x_1,x_2,\beta) dx_1\wedge dx_2+s(x_3,x_4,\beta)dx_3\wedge dx_4=s(x_1^{'},x_2^{'},\beta) dx_1^{'}\wedge dx_2^{'}+s(x_3^{'},x_4^{'},\beta)dx_3^{'}\wedge dx_4^{'},
\end{equation}
where the function $s(w,w_1,\alpha)$  is symmetric when we change $w\leftrightarrow w_1$ and is  given by
\[
s(w,w_1,\alpha)=\frac{\partial \psi(w,w_1,\alpha)}{\partial w_1}= \frac{\partial \psi(w_1,w,\alpha)}{\partial w}.
\]
\end{proposition}

\begin{proof}
It is known that the function $s$ is symmetric (see Lemma 9, \cite{ABS1}).  Differentiating three leg-forms centred at $x_1$ and $x_3$ and wedging with $dx_1$ and $dx_3$ respectively, we get
\begin{align*}
\frac{\partial \psi(x_1,x_2,\beta)}{\partial x_2} dx_1\wedge dx_2-\frac{\partial \psi(x_1,x_2^{'},\alpha)}{\partial x_2^{'}} dx_1\wedge dx_2^{'}&=\frac{\partial \phi(x_1,x_3,\alpha,\beta)}{\partial x_3} dx_1\wedge dx_3,\\
\frac{\partial \psi(x_3,x_4,\beta)}{\partial x_4} dx_3\wedge dx_4-\frac{\partial \psi(x_3,x_4^{'},\alpha)}{\partial x_4^{'}}dx_3\wedge dx_4^{'} &=\frac{\partial \phi(x_3,x_1,\alpha,\beta)}{\partial x_3} dx_3\wedge dx_1.
\end{align*}
Using Lemma 9 \cite{ABS1}, we get 
$\frac{\partial \phi(x_1,x_3,\alpha,\beta)}{\partial x_3}=\frac{\partial \phi(x_3,x_1,\alpha,\beta)}{\partial x_3}$. 
Therefore, we have
\[
s(x_1,x_2,\beta)dx_1\wedge dx_2-s(x_1,x_2^{'},\alpha)dx_1\wedge dx_2^{'}=-s(x_3,x_4,\beta)dx_3\wedge dx_4+s(x_3,x_4^{'},\alpha) dx_3\wedge dx_4{'},
\]
which implies
\begin{equation}
\label{E:identity1}
s(x_1,x_2,\beta)dx_1\wedge dx_2+s(x_3,x_4,\beta)dx_3\wedge dx_4=s(x_1,x_2^{'},\alpha)dx_1\wedge dx_2^{'}+s(x_3,x_4^{'},\alpha) dx_3\wedge dx_4{'}.
\end{equation}
Similarly, we get
\begin{equation}
\label{E:identity2}
s(x_1,x_2^{'},\alpha)dx_1\wedge dx_2^{'}+s(x_3,x_4^{'},\alpha) dx_3\wedge dx_4=s(x_1^{'},x_2^{'},\beta) dx_1^{'}\wedge dx_2^{'}+s(x_3^{'},x_4^{'},\beta)dx_3^{'}\wedge dx_4^{'}.
\end{equation}
Using \eqref{E:identity1} and \eqref{E:identity2}, we obtain \eqref{E:sympletic}.
\end{proof}

For equations  in the ABS list that require point transformations for the field variables and lattice parameters in order to bring their three-leg form to (\ref{E:3leg_form}), this proposition holds in the new variables. Next we list the function $s$ for all the cases of the ABS list.

\begin{itemize}
\item $H_1$:  $s(w,w_1,\alpha)=1$
\item $H_2$:  $s(w,w_1,\alpha)=\frac{1}{w+w_1+\alpha}$
\item $H_3^{\delta=0}$:   $s(w,w_1,\alpha)=\frac{1}{ww_1}$
\item $H_3^{\delta=1}$:   $s(w,w_1,\alpha)=\frac{1}{ww_1+\alpha}$

\item $Q_1^{\delta=0}$:
$ s(w,w_1,\alpha)=\frac{\alpha}{(w-w_1)^2}$

\item $Q_1^{\delta=1}$:
$s(w,w_1,\alpha)=-\frac{1}{w-w_1+\alpha}+\frac{1}{w-w_1-\alpha}$
\item $Q_2$:$s(w,w_1\alpha)=\frac{\alpha}{(w-w_1)^2-2 \alpha^2(w+w_1)+\alpha^4}$
\item
$Q_3^{\delta=0}$:
$s(w,w_1,\alpha)=\frac{\alpha^2-1}{(w_1-\alpha w)(w-\alpha w_1)}$
\item
$Q_3^{\delta=1}$:
$s(w,w_1\alpha)=N/D$
where
\begin{small}
\begin{align*}
N&=2\, \left(r(w)w+{w}^{2}-1 \right)  \left( r(w_1)w_{{1}}+{w_{{1}}}^{2}-1 \right)
 \left( w+
r(w) \right)  \left( w_{{1}}+r(w_1)\right) \alpha\, \left( {\alpha}^{2}-1 \right)
\\
D&=\Big( \alpha\,r(w_1)r(w)+\alpha\,r(w_1)w+
\alpha\,w_{{1}}r(w)+\alpha\,w_{{1}}w-1 \Big)
 \left( \alpha\,r(w)-r(w_1)+\alpha\,w-w_{{1}}
 \right)
 \\
 &\quad
  \left( r(w_1)r(w)+r(w_1)w+w_{{1}}r(w)+
w_{{1}}w-\alpha \right)
 \left( r(w)-\alpha r(w_1)-\alpha\,w_{{1}}+w \right) r(w)r(w_1)
\end{align*}
\end{small}
and $r(w)=\sqrt{w^2-1}$.
\item $Q_4$: we use the Hietarinta form \cite{Hie} for the $Q_4$ equation and we get
\begin{equation}
\label{Q4_Poi}
s(w,w_1,\alpha)=\frac{1}{\alpha^2 w^2 w_1^2+2 a ww_1+\alpha^2-w^2-w_1^2}, \ a=\sqrt{\alpha^4+\delta \alpha^2+1}.
\end{equation}
\end{itemize}

The 2-form $\omega_2=s(x_1,x_2,\beta) dx_1\wedge dx_2+s(x_3,x_4,\beta)dx_3\wedge dx_4$  is a symplectic form on the 
subset where $s(x_1,x_2,\beta) s(x_3,x_4,\beta) \neq 0$.  So, Proposition \ref{P:symplectics} implies that the maps derived by 
the $(2,2)$ periodic reduction of the ABS equations are symplectic with respect to it. Moreover, this proposition still holds when we change the edges $(x_1,x_2)$  and $(x_3,x_4)$ to $(x_4,x_1)$ and $(x_2,x_3)$ respectively and the parameter $\beta$ to $\alpha$. In other words,  
the map  $\mathcal{S}_{\alpha,\beta}$ is also symplectic with respect to $\omega_2'=s(x_2,x_3,\alpha) dx_2\wedge dx_3+s(x_4,x_1,\alpha)dx_4\wedge dx_2$.  

 \begin{remark} \label{rm3leg}
The 2-forms that were presented in \cite{ABS1} (Proposition 12) are obtained by adding these two symplectic structures, $\omega_2+\omega_2'$.
However it is important to note that these 2-forms are degenerate for  the cases of $H_1$ and $H_3^{\delta=0}$ and not symplectic as $\omega_2$ and $\omega_2'$. 
So they cannot convert to Poisson brackets for the associated maps by inverting their structure matrix.  However, Proposition \ref{P:symplectics} in this paper gives us a suitable Poisson structure for $(2,2)$ periodic reductions of these equations, as well as much simpler Poisson structure
for the non-degenerate cases. 
\end{remark}

\begin{corollary} \label{corPstr22}
The map \eqref{E:reduction} preserves the Poisson bracket  which is given by
\begin{equation} \label{ps22}
\frac{1}{s(x_1,x_2,\beta)}\frac{\partial}{\partial x_1} \wedge
\frac{\partial}{\partial
x_2}+\frac{1}{s(x_3,x_4,\beta)}\frac{\partial}{\partial x_3} \wedge
\frac{\partial}{\partial x_4}.
\end{equation}
\end{corollary}

We also note that  one can use the same technique to obtain symplectic forms for two component pKdV \cite{Hereman}. These forms will be reduced to the symplectic forms $\omega_2$ and $\omega_2'$ of pKdV if we set the two components to be equal.
On the other hand, the mutiparametric cross-ratio equation has a three-leg form but it does not give us  a similar result as  in Corollary~\ref{corPstr22}, since the corresponding function $s$ is not symmetric, i.e. $s(x_1,x_2,\alpha) \neq s(x_2,x_1,\alpha)$.

In addition, by considering the Poisson bracket of the $(2,2)$ periodic reductions we can derive a suitable Poisson structure for the corresponding lifts
of the quad-graph equations (as YB maps).  First we notice that the Poisson structure (\ref{ps22}) is equivalent with the bracket $\pi_2$ in (\ref{poisson3})
for $J_{\beta,\alpha}(x_1,x_2)=\frac{1}{s(x_1,x_2,\beta)}$. Then we can check if the corresponding YB maps $R_{\alpha,\beta}$ of lemma \ref{lemma1} are 
Poisson with respect to $\pi_1$ in (\ref{poisson1}). By direct calculation we can prove the next proposition.

\begin{proposition}
The lifts of the quad-graph equations $H_1,H_2,H_3,Q_1,Q_2,Q_3$, are
Poisson YB maps with respect to the Poisson structure $\pi_1$
described in (\ref{poisson1}), for
$J_{\alpha,\beta}(x_1,x_2)=\frac{1}{s(x_1,x_2,\alpha)}$.
\end{proposition}

Surprisingly, the lift of $Q_4$ with the corresponding structure
derived by \eqref{Q4_Poi} is not a Poisson map.

\section{Poisson structure on specific systems of quad-graph equations }

Lifts of 3D consistent quad-graph equations,
$w_2=F(w_1,w,w_{12},\alpha,\beta)$, are YB maps of the specific form
$R_{\alpha,\beta}(x_1,x_2,y_1,y_2) =
(F(y_1,x_1,y_2,\alpha,\beta),y_2,x_1,F(x_2,x_1,y_2,\alpha,\beta))$.
Next, we consider more general YB maps of the form
\begin{equation}
R_{\alpha,\beta}(x_1,x_2,y_1,y_2) = (F_1(y_1,x_1,y_2,\alpha,\beta),y_2,x_1,F_2(x_2,x_1,y_2,\alpha,\beta)) :=(u_1,u_2,v_1,v_2), \label{Rsys}
\end{equation}
that involve two different functions $F_1$ and $F_2$.
As we are going to show, these maps are related to a specific kind of systems of quadrilateral equations. As before, suitable Poisson structures of
the YB maps give rise to suitable Poisson structure of specific periodic reductions of the system.  We summarise our results in the next proposition.

\begin{proposition}  \label{propSys}
Let $R_{\alpha,\beta}$, defined by (\ref{Rsys}), be a YB map with Lax matrix $L$. Then \\
$\alpha)$  the system of equations
\begin{equation} \label{sysF1F2}
w_2=F_2(w_1,t,w_{12},\alpha,\beta), \ t_2=F_1(t_1,t,w_{12},\alpha,\beta),
\end{equation}
satisfies the Lax representation
\begin{equation} \label{laxsys}
L(t_2,w_{12},\alpha)L(t,w_2,\beta)=L(t_1,w_{12},\beta)L(t,w_1,\alpha).
 \end{equation}
 $\beta)$ The $(1,1)$ periodic reduction $\mathcal{S}_{\alpha,\beta}((x_1,y_1),(x_2,y_2))=R_{\alpha,\beta}  \circ R_{\alpha,\beta}(y_1,x_2,y_2,x_1)$. \\
 $\gamma)$ If $R_{\alpha,\beta}$ is Poisson with respect to the Sklyanin bracket (\ref{extskl}), then the $(1,1)$ periodic reduction
 $\mathcal{S}_{\alpha,\beta}$ of the system (\ref{sysF1F2}) is Poisson with respect to the bracket
 \begin{eqnarray*}
\{L(y_1,x_2,\alpha,\zeta ) \ \overset{\otimes }{,} \
L(y_1,x_2,\alpha ,h)\}&=&[\mathbf{r}(\zeta,h),L(y_1,x_2,\alpha,\zeta ) \otimes
L(y_1,x_2,\alpha,h)], \nonumber  \\
\{L(y_2,x_1,\beta,\zeta ) \ \overset{\otimes }{,} \
L(y_2,x_1,\beta ,h)\}&=&[\mathbf{r}(\zeta,h),L(y_2,x_1,\beta,\zeta ) \otimes
L(y_2,x_1,\beta,h)]
\end{eqnarray*}
and $\{ x_i,x_j \}=\{ y_i,y_j \}=\{ x_1,y_1 \}=\{ x_2,y_2 \}=0$.

 \end{proposition}

\begin{proof}
$\alpha$) From the Lax representation of the YB map $R_{\alpha,\beta}$ we have
\begin{equation} \label{prosyst}
L(u_1,y_2,\alpha)L(x_1,v_2,\beta)=L(y_1,y_2,\beta)L(x_1,x_2,\alpha).
\end{equation}
Also, from (\ref{Rsys}, \ref{sysF1F2}), by setting $x_1=t$,
$x_2=w_1$, $y_1=t_1$ and $y_2=w_{12}$, we derive $u_1=t_2$ and
$v_2=w_2$. If we substitute these values to (\ref{prosyst}),
 we derive the Lax representation (\ref{laxsys}).
\\
$\beta)$ We consider the $(1,1)$ periodic reduction of the system (\ref{sysF1F2}), with initial values $(x_1,y_1)$, $(x_2,y_2)$ as in Fig. \ref{F:system},
$\mathcal{S}_{\alpha,\beta}:((x_1,y_1),(x_2,y_2)) \mapsto ((x_1',y_1'),(x_2',y_2'))$. Then from (\ref{sysF1F2}) we have \\
$(y_1',x_2',y_2',x_1')=(F_1(y_1,y_2',x_2',\alpha,\beta),F_2(x_2,y_1,x_1,\alpha,\beta),F_1(y_2,y_1,x_1,\alpha,\beta),F_2(x_2,y_2',x_2',\alpha,\beta))$
which is equal to $R_{\alpha,\beta}  \circ R_{\alpha,\beta}(y_1,x_2,y_2,x_1).$ \\
$\gamma)$ It follows directly from $\alpha)$ and $\beta)$.
 \end{proof}

\begin{figure}
\begin{subfigure}{.5\textwidth}
\centertexdraw{\setunitscale 0.38 \linewd 0.03 \arrowheadtype t:F
\move (-1 -1) \lvec (-1 1) \lvec (1 1) \lvec (1. -1) \lvec(-1 -1)
\move(-1 -1) \fcir f:0.0 r:0.1 \move(-1 1) \fcir f:0.0 r:0.1 \move(1
1) \fcir f:0.0 r:0.1 \move(1 -1) \fcir f:0.0 r:0.1 \htext (-2.2
-1.2) {$(w,t)$} \htext (-2.4 1.) {$(w_2,t_2)$} \htext (1.2 1) {$(w_{12},t_{12})$}
\htext (1.2 -1.2) {$(w_{1},t_1)$} \htext (-0.05 1.2) {$\alpha$} \htext
(-0.05 -1.4) {$\alpha$} \htext (-1.5 -0.15) {$\beta$} \htext (1.2
-0.15) {$\beta$} }
\end{subfigure}
\begin{subfigure}{.5\textwidth}
\begin{center}
\begin{tikzpicture}[line cap=round,line join=round,>=triangle 45,x=1.5cm,y=1.5cm][h]
\draw (0.0,1.0)-- (1.0,0.0);
\draw (1.0,0.0)-- (2.0,1.0);
\draw (2.0,1.0)-- (3.0,0.0);
\draw (3.0,0.0)-- (4.0,1.0);

\draw (0.0,1.0)-- (1.0,2.0);
\draw (1.0,2.0)-- (2.0,1.0);
\draw (2.0,1.0)-- (3.0,2.0);
\draw (3.0,2.0)-- (4.0,1.0);

\draw (1.0,2.0)-- (2.0,3.0);
\draw (2.0,3.0)-- (3.0,2.0);
\draw [->] (1.0,0.0) -- (1.0,2.0);
\draw [->] (2.0,1.0) -- (2.0,3.0);

\draw [fill=black] (0.0,1.0) circle (2pt);
\draw(-0.2,1.2585663160937274) node {\scriptsize{$(x_1,y_1)$}};
\draw [fill=black] (1.0,0.0) circle (2pt);
\draw(1.0115929156639173,-0.21046306447702798) node {\scriptsize{$(x_2,y_2)$}};
\draw (0.39235541979040217,0.5043129347175777) node {$\alpha$};
\draw [fill=black] (2.0,1.0) circle (2pt);
\draw (2.0005444724900006,1.3085663160937274) node {\scriptsize{$(x_1,y_1)$}};
\draw (1.7109574955585132,0.4593605912254829) node {$\beta$};
\draw [fill=black] (3.0,0.0) circle (2pt);
\draw (3.1044801438134484,-0.21046306447702798) node {$x_2$};
\draw(2.3702585334425685,0.4593605912254829) node {$\alpha$};
\draw [fill=black] (4.0,1.0) circle (2pt);

\draw [fill=black] (1.0,2.0) circle (2pt);
\draw (0.81115929156639173,2.2525019874171747) node {\scriptsize{$(x_2',y_2')$}};
\draw (0.707021824235065,1.4632962625489305) node {$\beta$};
\draw (1.3663228621191206,1.4632962625489305) node {$\alpha$};
\draw [fill=black] (3.0,2.0) circle (2pt);
\draw (3.21044801438134484,2.2525019874171747) node{\scriptsize{$(x_2',y_2')$}};
\draw (2.6999090523845966,1.4632962625489305) node {$\beta$};
\draw(3.3741942047660167,1.4632962625489305) node {$\alpha$};

\draw [fill=black] (4.0,1.0) circle (2pt);
\draw(4.158415815136897,1.2585663160937274) node{\scriptsize{$(x_1,y_1)$}};
\draw (3.7038447237080443,0.4593605912254829) node {$\beta$};

\draw [fill=black] (2.0,3.0) circle (2pt);
\draw (2.0305444724900006,3.216437658740622) node {\scriptsize{$(x_1',y_1')$}};
\draw(1.7109574955585132,2.452247819375013) node {$\beta$};
\draw(2.3702585334425685,2.452247819375013) node {$\alpha$};

\end{tikzpicture}
\end{center}
\end{subfigure}
\caption{The system of equations (\ref{sysF1F2}) at the vertices of a quadrilateral and the $(1,1)$ periodic reduction \label{F:system}}
\end{figure}
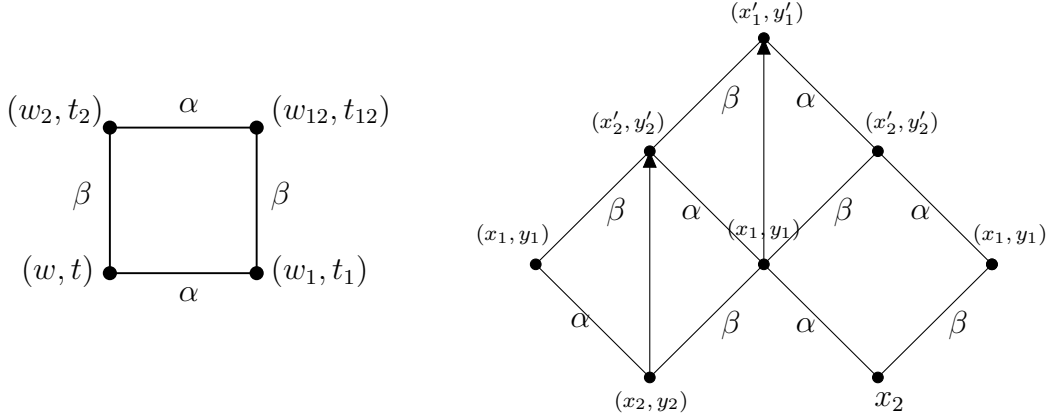

Equivalently, one can start by considering 3D consistent systems of the special form (\ref{sysF1F2}) and derive YB maps as
(\ref{Rsys}) and if the $(1,1)$ periodic reduction of these systems is Poisson, then the composition $R_{\alpha,\beta}  \circ R_{\alpha,\beta} $
will be Poisson too.

\subsubsection*{Multiparametric lattice NLS system}

We will apply the last results to a multiparametric form of the
lattice nonlinear Schr{\"o}dinger system (NLS) that is related to a
generalization of the Adler-Yamilov map.

We begin with the parametric map
 $R_{\alpha,\beta}^\varepsilon:(x_1,x_2,y_1,y_2) \mapsto (u_1,u_2,v_1,v_2)$, with
\begin{eqnarray}
{u}_{1} =\frac{\alpha_2}{\beta_2} y_{1}-%
\frac{\varepsilon^2 ( \alpha_1\beta_2-\alpha_2 \beta_{1})}{ \beta_2 (\alpha_2 \beta_2+ \varepsilon^2 x_1 y_2)}
x_{1},
 {u}_{2}=y_{2}, 
{v}_{1} =x_{1},   {v}%
_{2}=\frac{\beta_2}{\alpha_2}x_{2}+\frac{ \varepsilon^2 (\alpha_1\beta_2-\alpha_2
\beta_{1})}{\alpha_2(\alpha_2 \beta_2+  \varepsilon^2 x_1 y_2)} y_2 \label{gAYmap}
\end{eqnarray}
For any $\varepsilon$, this map is a YB map with 
vector parameters $\alpha=(\alpha_1,\alpha_2)$,
$\beta=(\beta_{1},\beta_2)$ and strong Lax matrix
\begin{equation}
\label{E:Lax_NLS}
{L}(x_{1},x_{2};\alpha)=
\begin{pmatrix}
\frac{\varepsilon}{\alpha_2}(\alpha_1+x_{1}x_{2})-\varepsilon \zeta  &
x_{1} \\
x_{2} & \frac{\alpha_2}{\varepsilon}
\end{pmatrix}.
\end{equation}
In this case equations (\ref{extskl}) are equivalent with the coordinate Poisson brackets
\begin{equation} \label{pbrAY}
\{ x_1,x_2 \}=\alpha_2, \ \{ y_1,y_2 \}=\beta_2, \ \{ x_i,y_j \}=0,
\end{equation}
and the map $R_{\alpha,\beta}$ is Poisson with respect to this bracket.
The corresponding functions $F_1$, $F_2$ are defined by the relations  $F_1(y_1,x_1,y_2,\alpha,\beta)=u_1$ and
$F_2(x_2,x_1,y_2,\alpha,\beta)=v_2$.

So, according to Proposition (\ref{propSys}), the system (\ref{sysF1F2})  satisfies the Lax
representation $L(t_2,w_{12},\alpha)L(t,w_2,\beta)=L(t_1,w_{12},\beta)L(t,w_1,\alpha)$. We can write this system as
 \begin{eqnarray}
&&\alpha_2^2 \beta_2 t_1-\beta_2^2  \alpha_2 t_2- \varepsilon^2 t((\beta_2 w_1- \alpha_2 w_2)t+\alpha_1 \beta_2-\alpha_2 \beta_{1}) = 0 , \nonumber \\
&&\beta_2^2  \alpha_2 w_1-\alpha_2^2  \beta_2 w_2 +\varepsilon^2 w_{12}(( \beta_2
w_1- \alpha_2 w_2)t+\alpha_1 \beta_2-\alpha_2 \beta_{1}) = 0.
\label{gsystNLS}
\end{eqnarray}
The $(1,1)$ periodic reduction of this system
$\mathcal{S}_{\alpha,\beta}((x_1,y_1),(x_2,y_2))=R_{\alpha,\beta}  \circ R_{\alpha,\beta}(y_1,x_2,y_2,x_1)$
 is a Poisson map with respect to
\begin{equation} \label{pbrAY2}
\{ y_1,x_2 \}=\alpha_2, \ \{ y_2,x_1 \}=\beta_2, \ \{ x_i,x_j \}=\{
y_i,y_j \}=\{ x_1,y_1 \}=\{ x_2,y_2 \}=0.
\end{equation}

The YB map (\ref{gAYmap}) is a special case\footnote{It is derived
by setting $\alpha_1=\beta_1=\varepsilon$ and renaming the parameters $\alpha_2$ and $\alpha_3$  into $\alpha_1$ and $\alpha_2$ respectively 
at the
map presented in \cite{kp3}.} of a generalisation of the
Adler-Yamilov map that was presented in \cite{kp3}. We have checked
that this lattice system is not 3D consistent. In fact, given
initial values $t,t_1,w_1,w_2,w_3$  and parameters
$(\alpha_1,\alpha_2), (\beta_1,\beta_2), (\gamma_1,\gamma_2)$ on the
cube, one has three ways of computing $w_{123}$ and gets different
values (two of the three values which are associated with $w_{13}$ are
the same). However, it still gives us the Lax pair given
by~\eqref{E:Lax_NLS}.

For $\varepsilon=\alpha_2=\beta_2=1$, the map (\ref{gAYmap}) is reduced to the
Adler-Yamilov YB map \cite{adlery} and the system (\ref{gsystNLS})
to the lattice NLS system \cite{Hereman}. In this case the
corresponding symplectic structures of the YB map and the $(1,1)$
reduction becomes $dx_1 \wedge dx_2+dy_1 \wedge dy_2$ and $dy_1
\wedge dx_2+dy_2 \wedge dx_1$ respectively.

\section{Integrability}

The Lax representation of quadrilateral equations and YB maps
provides integrals of periodic reductions. In the case of the lifts
of quad-graph equations as YB maps, the corresponding integrals are
derived from the trace of the monodromy matrix, 
\begin{equation} \label{mon1}
M_1(x_1,x_2,y_1,y_2)=L(y_1,y_2,\alpha)L(x_1,x_2,\beta),
\end{equation} 
while in the case of $(2,2)$ periodic reductions from the trace of
\begin{equation} \label{mon2}
M_2(x_1,x_2,x_3,x_4)=L(x_2,x_1,\beta)
L(x_3,x_2,\alpha)L(x_4,x_3,\beta)L(x_1,x_4,\alpha).
\end{equation}

If the Poisson structure is derived from the Sklyanin bracket, the
involutivity of the integrals follows directly 
\cite{bab}. In the rest of the cases we have to check it by computing their
brackets.  

The trace of the monodromy
matrix (\ref{mon2}) for  $H_1$ and $Q_1^{\delta=0}$ does not give us enough
integrals. Two additional integrals for these equations are
$I_{H_1}:=x_1+x_2+x_3+x_4$ and
$I_{Q_1^{\delta=0}}:=\frac{(x_1-x_2)(x_3-x_4)}{(x_2-x_3)(x_1-x_4)}$, respectively. 
In the case of the corresponding lifts, the trace of (\ref{mon1}) gives enough integrals 
for  $Q_1^{\delta=0}$ but only one for $H_1$. One extra integral for the lift of $H_1$ is 
$J_{H_1}:=(x_1+x_2-y_1-y_2)^2$.  Finally, all the necessary integrals of the multilinear 
cross-ratio equation and NLS system (as YB maps and periodic reductions) are derived from the trace 
of their monodromy matrices.

The integrals of all the examples that have been presented in this
paper are in involution with respect to the corresponding Poisson 
structures.

Higher dimensional maps can be derived by considering initial value
problems with greater periodicity. The YB maps that we studied here
as lifts of quadrilateral equations are four dimensional maps and
can be considered as the transfer maps of one periodic initial value
problem (see e.g. \cite{kp2}), with two fields on each edge of every
elementary quadrilateral. We can consider $n$-periodic initial value
problem to derive $4n$ dimensional transfer maps. By extending the
Poisson bracket to $\mathbb{C}^{4n}$, in the natural way, the
corresponding transfer maps will be Poisson. In this case the
integrals are derived from the trace of the monodromy matrix
$$M_n(x_1,\ldots,x_{2n},y_1, \ldots,
y_{2n})=L(y_{2n},y_{2n-1},\beta)L(x_{2n},x_{2n-1},\alpha) \ldots
L(x_2,x_1,\alpha)L(y_2,y_1,\beta).$$ As before, if the Poisson
structure is coming from the (extended) Sklyanin bracket, they will
be in involution.

On the other hand, in the case of $(n,n)$ periodic reductions of the
quad-graph equations with fields on the vertices we cannot extend the
corresponding Poisson bracket (\ref{ps22}) in a similar way and
derive Poisson maps. Therefore our analysis in this paper stops at
$(2,2)$ periodic reductions. There might be a way of extending the
presented Poisson structure to some kind of higher periodic
reductions and this is an issue that we would like to investigate in
the near future. However, it is noted that due to the symmetry of
Lagrangians of the ABS equations, one can obtain a
presymplectic structure for $(n,n)$ periodic reductions. This is a more
complicated structure that involves all the short legs and is given
as follows cf.\cite{ABS1}
\begin{equation*}
s(x_1,x_2,\beta) dx_1\wedge dx_2+s(x_2,x_3,\alpha) dx_2\wedge
dx_3+s(x_3,x_4,\beta)dx_3\wedge dx_4+\ldots +s(x_{2n},x_1,\alpha)
dx_1\wedge dx_2.
\end{equation*}
This form is degenerate in the cases of $H_1$ and $Q_1^{\delta=0}$ (remark \ref{rm3leg}).
A similar form can be derived for the staircase periodic reductions of the multiparametric cross-ratio equation (\ref{genQ1}).

\section{Conclusion}

We presented two different ways to obtain Poisson structures that
are preserved under the four dimensional maps derived as lifts and
periodic reductions of integrable lattice equations. The
corresponding integrals are in involution with respect to them.

We applied our results in the cases that can be viewed in the
following table. The check mark $\times$ refer to Poisson structures
derived from the Sklyanin bracket, while the \checkmark to
structures derived from three-leg forms. The last two columns refer to
multiparametric cross-ration equation and NLS system respectively. 
\\ \\ 
\begin{tabular}{ |p{2.9cm}||p{0.7cm}|p{0.7cm}|p{0.7cm} |p{0.97cm}|p{0.97cm}|p{0.7cm}|p{0.7cm}|p{0.7cm}|p{1.2cm}|p{1.35cm}|  }
 \hline
 \hline
\ \ & \ $H_1$ & \ $H_2$ & \ $H_3$ & \ $Q_1^{\delta=0}$ & \ $Q_1^{\delta=1}$ & \ $Q_2$ & \ $Q_3$ & \ $Q_4$ & \ m.c-r & \ m.NLS \\
 \hline
 Lifts (YB maps)  & $\times$\checkmark  & \ \ \checkmark &  \ \ \checkmark & $ \ \times$\checkmark  & \ \ \ \checkmark &  \ \ \checkmark  & \ \  \checkmark   & \ &   \ \ \ $\times$ & \ \ \ $\times$ \\
 Per. reductions & $\times$\checkmark & \ \ \checkmark &  \ \ \checkmark & $ \ \times$\checkmark & \ \ \ \checkmark & \  \ \checkmark  &\ \  \checkmark  & \ \checkmark  & \ &
\ \ \ $\times$
  \\
\hline \hline
\end{tabular}
\\ \\ 

Regarding the Sklyanin bracket, we have only been able to check a few of the Lax
representations of the ABS list. By considering different Lax pairs
it might be possible to convert more cases in this framework. In
the cases that we can apply both ways like in $H_1$ and
$Q_1^{\delta=0}$, we derive the same Poisson structure from the
Sklyanin bracket and from the corresponding three-leg form. This
fact suggests a deeper relation between these two different
approaches that we would like to investigate in the future.

Finally, we believe that the study of multiparametric extensions of
known integrable lattice equations and YB maps, as well as the
significance of the extra parameters on the continuum limits is an
interesting issue that deserves further attention.

\section*{Acknowledgement}
This research was supported by the Australian Research Council.  The authors would like to thank Prof Reinout Quispel for his useful comments.

\appendix
\section{\\Proof of Lemma \ref{lemma1}} \label{App:AppendixA}
Let
$\mathbf{S}_{\alpha,\beta}(x_1,x_2,y_1,y_2)=(U_1,U_2,V_1,V_2)$. From
(\ref{R}, \ref{S}, \ref{T}), we have
$$u_1=u_1'=V_1, \ u_2=u_2'=\tilde{u}_2=y_2, \ v_1=v_1'=\tilde{v}_1=x_1, \ v_2=\tilde{v}_2 \ \text{and} \ U_2=v_2'.$$
If $\mathbf{S}_{\alpha,\beta}$ is Poison map with respect to
$\pi_2$, then $\pi_2(dU_2,dV_1)=\pi_2(dv_2',du_1')=0$. From the last
equality and from (\ref{poisson2}) we derive (\ref{condition}).

On the other hand, let us suppose that (\ref{condition}) holds.
Since $R_{\alpha,\beta}$ is a Poisson map with respect to $\pi_1$,
$\pi_1(du_1,du_2)=J_{\alpha,\beta}(u_1,u_2)$ and
$\pi_1(dv_1,dv_2)=J_{\beta,\alpha}(v_1,v_2)$. From these relations, using  (\ref{poisson1}), we
derive
\begin{eqnarray}
\frac{\partial u'_1}{\partial{y_1}}J_{\beta,\alpha}(y_1,y_2)=\frac{\partial u_1}{\partial{y_1}}J_{\beta,\alpha}(y_1,y_2)= J_{\alpha,\beta}(u_1,u_2)=J_{\alpha,\beta}(u_1',y_2)=J_{\alpha,\beta}(u_1',u_2') \label{proof1} \\
\frac{\partial \tilde{v}_2}{\partial{x_2}}J_{\alpha,\beta}(x_1,x_2)=\frac{\partial {v}_2}{\partial{x_2}}J_{\alpha,\beta}(x_1,x_2)= J_{\beta,\alpha}(v_1,v_2)=J_{\beta,\alpha}(x_1,\tilde{v}_2)=J_{\beta,\alpha}(\tilde{v}_1,\tilde{v}_2) \label{proof2}
\end{eqnarray}
Equivalently, by interchanging the parameters $\alpha$ and $\beta$ in the last two equations, we derive from (\ref{D4sy})
\begin{eqnarray}
\frac{\partial \tilde{u}_1}{\partial{y_1}}J_{\alpha,\beta}(y_1,y_2)=J_{\beta,\alpha}(\tilde{u}_1,y_2)=J_{\beta,\alpha}(\tilde{u}_1,\tilde{u}_2), \label{proof3}  \\
\frac{\partial {v_2'}}{\partial{x_2}}J_{\beta,\alpha}(x_1,x_2)=J_{\alpha,\beta}(x_1,v_2')=J_{\alpha,\beta}(v_1',v_2'). \label{proof4}
\end{eqnarray}
So, from (\ref{proof3}) and (\ref{proof4}) we have that
\begin{equation} \label{proof5}
J_{\beta,\alpha}(U_1,U_2)= \frac{\partial U_1}{\partial v_1'} J_{\alpha,\beta}(v_1',v_2')=
\frac{\partial U_1}{\partial v_1'} \frac{\partial v_2'}{\partial x_2}J_{\alpha,\beta}(x_1,x_2).
\end{equation}
Similarly, from (\ref{proof2},\ref{proof1})
\begin{equation} \label{proof6}
J_{\beta,\alpha}(V_1,V_2)= \frac{\partial V_2}{\partial u_2'} J_{\alpha,\beta}(u_1',u_2')=
\frac{\partial V_2}{\partial u_2'} \frac{\partial u_1'}{\partial y_1}J_{\beta,\alpha}(y_1,y_2).
\end{equation}
Furthermore, condition (\ref{condition}) is equivalent to
\begin{equation} \label{condition2}
J_{\alpha,\beta}(x_1,x_2)\frac{\partial \tilde{u}_1}{\partial x_1}  \frac{\partial \tilde{v}_2}{\partial x_2}=J_{\alpha,\beta}(y_1,y_2)\frac{\partial \tilde{u}_1}{\partial y_1}  \frac{\partial \tilde{v}_2}{\partial y_2}.
\end{equation}

Next, we calculate the Poisson brackets  $\pi_2(dU_1,dU_2)$,  $\pi_2(dV_1,dV_2)$ and  $ \pi_2(dU_i,dV_j)$, $i,j=1,2$. From (\ref{poisson2}) 
we get 
\begin{eqnarray}
\pi_2(dU_1,dU_2) &=&J_{\beta,\alpha}(x_1,x_2)\frac{\partial{U_1}}{\partial{v_1'}}\frac{\partial {v_2'}}{\partial{x_2}}+\frac{\partial{U_1}}{\partial{u_1'}}(J_{\beta,\alpha}(x_1,x_2)\frac{\partial u_1'}{\partial x_1}  \frac{\partial v_2'}{\partial x_2}-J_{\beta,\alpha}(y_1,y_2)\frac{\partial u_1'}{\partial y_1}  \frac{\partial v_2'}{\partial y_2})  \nonumber \\
&=&J_{\beta,\alpha}(U_1,U_2)+\frac{\partial{U_1}}{\partial{u_1'}}(J_{\beta,\alpha}(x_1,x_2)\frac{\partial u_1'}{\partial x_1}  \frac{\partial v_2'}{\partial x_2}-J_{\beta,\alpha}(y_1,y_2)\frac{\partial u_1'}{\partial y_1}  \frac{\partial v_2'}{\partial y_2}), \label{proof7} \\
\pi_2(dV_1,dV_2) &=&-J_{\beta,\alpha}(y_1,y_2)\frac{\partial{V_2}}{\partial{u_2'}}\frac{\partial {u_1'}}{\partial{y_1}}+\frac{\partial{V_2}}{\partial{v_2'}}(J_{\beta,\alpha}(x_1,x_2)\frac{\partial u_1'}{\partial x_1}  \frac{\partial v_2'}{\partial x_2}-J_{\beta,\alpha}(y_1,y_2)\frac{\partial u_1'}{\partial y_1}  \frac{\partial v_2'}{\partial y_2})  \nonumber \\
&=&-J_{\beta,\alpha}(V_1,V_2)+\frac{\partial{V_2}}{\partial{u_2'}}(J_{\beta,\alpha}(x_1,x_2)\frac{\partial u_1'}{\partial x_1}  \frac{\partial v_2'}{\partial x_2}-J_{\beta,\alpha}(y_1,y_2)\frac{\partial u_1'}{\partial y_1}  \frac{\partial v_2'}{\partial y_2}), \label{proof8}
\end{eqnarray}
where for the last equalities of (\ref{proof7}) and (\ref{proof8}) we used (\ref{proof5}) and (\ref{proof6}) respectively.
Also,
\begin{eqnarray} \label{proof9}
\pi_2(dU_1,dV_1) &=&-\frac{\partial{U_1}}{\partial{v_2'}}(J_{\beta,\alpha}(x_1,x_2)\frac{\partial u_1'}{\partial x_1}  \frac{\partial v_2'}{\partial x_2}-J_{\beta,\alpha}(y_1,y_2)\frac{\partial u_1'}{\partial y_1}  \frac{\partial v_2'}{\partial y_2}), \\
\pi_2(dU_2,dV_2) &=&-\frac{\partial{V_2}}{\partial{u_1'}}(J_{\beta,\alpha}(x_1,x_2)\frac{\partial u_1'}{\partial x_1}  \frac{\partial v_2'}{\partial x_2}-J_{\beta,\alpha}(y_1,y_2)\frac{\partial u_1'}{\partial y_1}  \frac{\partial v_2'}{\partial y_2}),  \label{proof10} \\
\pi_2(dU_2,dV_1) &=&-(J_{\beta,\alpha}(x_1,x_2)\frac{\partial u_1'}{\partial x_1}  \frac{\partial v_2'}{\partial x_2}-J_{\beta,\alpha}(y_1,y_2)\frac{\partial u_1'}{\partial y_1}  \frac{\partial v_2'}{\partial y_2}) \label{proof11} \end{eqnarray}
and
\begin{eqnarray*}
\pi_2(dU_1,dV_2) &=&(\frac{\partial{U_1}}{\partial{u_1'}}\frac{\partial{V_2}}{\partial{v_2'}}-\frac{\partial{U_1}}{\partial{v_2'}}\frac{\partial{V_2}}{\partial{u_1'}})(J_{\beta,\alpha}(x_1,x_2)\frac{\partial u_1'}{\partial x_1}  \frac{\partial v_2'}{\partial x_2}-J_{\beta,\alpha}(y_1,y_2)\frac{\partial u_1'}{\partial y_1}  \frac{\partial v_2'}{\partial y_2}) \nonumber \\
&+&\frac{\partial v_2'}{\partial x_2} J_{\beta,\alpha}(x_1,x_2)\frac{\partial U_1}{\partial v_1'}  \frac{\partial V_2}{\partial v_2'}-\frac{\partial u_1'}{\partial y_1}J_{\beta,\alpha}(y_1,y_2)\frac{\partial U_1}{\partial u_1'}  \frac{\partial V_2}{\partial u_2'}
\end{eqnarray*}
or from (\ref{proof1},\ref{proof4})
\begin{eqnarray}
\pi_2(dU_1,dV_2) &=&(\frac{\partial{U_1}}{\partial{u_1'}}\frac{\partial{V_2}}{\partial{v_2'}}-\frac{\partial{U_1}}{\partial{v_2'}}\frac{\partial{V_2}}{\partial{u_1'}})(J_{\beta,\alpha}(x_1,x_2)\frac{\partial u_1'}{\partial x_1}  \frac{\partial v_2'}{\partial x_2}-J_{\beta,\alpha}(y_1,y_2)\frac{\partial u_1'}{\partial y_1}  \frac{\partial v_2'}{\partial y_2}) \nonumber \\
&+& J_{\alpha,\beta}(v_1',v_2')\frac{\partial U_1}{\partial v_1'}  \frac{\partial V_2}{\partial v_2'}-J_{\alpha,\beta}(u_1',u_2')\frac{\partial U_1}{\partial u_1'}  \frac{\partial V_2}{\partial u_2'}. \label{proof12}
\end{eqnarray}

From conditions (\ref{condition}) and (\ref{condition2}), (\ref{proof7}, \ref{proof8}, \ref{proof9}, \ref{proof10}, \ref{proof11}, \ref{proof12})
become
$$\pi_2(dU_1,dU_2)=J_{\beta,\alpha}(U_1,U_2), \ \pi_2(dV_1,dV_2)=-J_{\beta,\alpha}(V_1,V_2), \   \pi_2(dU_i,dV_j)=0, \  \text{for} \ i,j=1,2,$$
i.e. the map $\mathbf{S}_{\alpha,\beta}$ is Poisson.

\end{document}